\documentclass[twoside]{article}

\usepackage{amsmath,amsthm,amssymb}

\usepackage{newtxtext,newtxmath}

\usepackage[text={12.5cm,19cm},centering,paperwidth=17cm,paperheight=24cm]{geometry}

\usepackage[fontsize=10.4pt]{scrextend}

\def\titlerunning#1{\gdef\titrun{#1}}

\makeatletter
\def\author#1{\gdef\autrun{\def\and{\unskip, }#1}\gdef\@author{#1}}
\def\address#1{{\def\and{\\\hspace*{15.6pt}}\renewcommand{\thefootnote}{}\footnote{#1}}\markboth{\autrun}{\titrun}}
\makeatother

\def\email#1{email: \href{mailto:#1}{#1} }
\def\subjclass#1{\par\bigskip\noindent\textbf{Mathematics Subject Classification 2020.} #1}
\def\keywords#1{\par\smallskip\noindent\textbf{Keywords.} #1}

\newenvironment{dedication}{\itshape\center}{\par\medskip}
\newenvironment{acknowledgments}{\bigskip\small\noindent\textit{Acknowledgments.}}{\par}

\newtheorem{thm}{Theorem}[section]

\newtheorem{lem}[thm]{Lemma}

\theoremstyle{definition}

\newtheorem*{rem}{Remark}

\numberwithin{equation}{section}

\usepackage{bbm}
\usepackage{csquotes}
\DeclareMathOperator{\Spec}{\mathrm{Spec}}
\DeclareMathOperator{\supp}{\mathrm{supp}}
\newcommand{\one}{\mathbbm{1}}
\newcommand{\rmu}{\rho_{\mu}}
\newcommand{\Delper}{\Delta^{{\rm per}}}
\newtheorem{assumption}[thm]{Assumption}

\usepackage[hyperfootnotes=false,colorlinks=true,allcolors=blue]{hyperref}

\begin{document}

\titlerunning{Length scales for BEC in the dilute Bose gas}

\title{\textbf{Length scales for BEC in the dilute Bose gas}}

\author{S{\o}ren Fournais}

\date{}

\maketitle

\address{S. Fournais, Aarhus University, Ny Munkegade 181, DK-8000 Aarhus C, Denmark; \email{fournais@math.au.dk}, and Institute for Advanced Study, 1 Einstein Drive, 08540 Princeton, USA; \email{fournais@ias.edu}}

\begin{dedication}
To Ari Laptev on the occasion of his 70th birthday
\end{dedication}

\begin{abstract}
We give a short proof of Bose Einstein Condensation of dilute Bose gases on length scales much longer than the Gross-Pitaevskii scale.
\end{abstract}


\section{Introduction}
In this article, I will review how recent techniques from \cite{BFS,FS} can be used to give simple proofs of Bose-Einstein Condensation (BEC) on different length scales for the dilute Bose gas.

One of the outstanding open problems in mathematical physics is to prove BEC in the thermodynamic limit for a continuous,  (weakly) interacting system\footnote{See \url{http://web.math.princeton.edu/~aizenman/OpenProblems_MathPhys/} for information on this and other challenging problems in Mathematical Physics.}. Unfortunately, we will have nothing new to report on this question. However, the substantially easier problem of proving BEC on shorter length scales has seen marked progress in recent years \cite{LS,ABS,BBCS0,BBCS2}.
The first paper cited gave the first proof of BEC on the so-called Gross-Pitaevskii length scale, which will be defined precisely below. On this length scale the energy gap between the ground state of the kinetic energy and the first excited state is of the same order of magnitude as the energy per particle in the system. Therefore it is feasible, but still not easy, to prove condensation on this length scale. 

The recent papers \cite{BS,BFS,FS} are dedicated to the study of the energy asymptotics of the Bose gas. By adapting the methods of these papers, one obtains easy proofs of condensation on length scales substantially longer than the Gross-Pitaevskii scale (but far shorter than \enquote{thermodynamic length scales}). The reference \cite{ABS} goes beyond the Gross-Pitaevskii length scale but with a very different method than ours and for a smaller range of lengths.

For background on the mathematics of the Bose gas we recommend the reviews \cite{LSSY,Rou} and references therein.

We proceed to give the precise statement of the result in the present paper.
We study the following Hamiltonian $H= H(N, L)$,
\begin{align}
H = \sum_{j=1}^N - \Delper_j + \sum_{j < k } v^{\rm per}(x_j-x_k).
\end{align}
Here $\Delper$ is the Laplace operator on $\Omega = (-L/2, L/2)^3$ with periodic boundary conditions and 
$v^{\rm per}= \sum_{j \in  {\mathbb Z}^3} v(x-Lj)$ is the periodic version of the potential $v$.
In other words, we identify $\Omega$ with the torus ${\mathbb R}^3/L {\mathbb Z}^3$ and consider the Laplacian and the potential $v$ as given on the torus.

The ground state energy $E = E(N, L)$ of $H$ is defined to be
\begin{align}
E(N, L) := \inf \Spec H(N, L).
\end{align}
The particle density in the box is $\rho:=N/L^3$.

Define furthermore,
\begin{align}
P_{\Omega} := L^{-3} | 1 \rangle \langle 1|, \qquad Q_{\Omega}  := \one - P_{\Omega},
\end{align}
i.e. $P_{\Omega} $ is the orthogonal projection in $L^2(\Omega)$ onto the constant functions and $Q_{\Omega} $ is the orthogonal projection to the complement.
Using these definitions, we define
\begin{align}\label{eq:SF_n0ognplus}
n_0:= \sum_{j=1}^N P_{\Omega,j}, \qquad n_{+} := \sum_{j=1}^N Q_{\Omega,j} = N - n_0.
\end{align}
The definition is based on the idea that low energy eigenstates of the system should concentrate in the constant function. Thus, $n_0$ counts the number of particles \enquote{in the condensate} and $n_{+}$ the number of particles \enquote{excited out of the condensate}.

That the ground state $\Psi$ of $H(N,\Omega)$ exhibits Bose-Einstein condensation in the constant function means by definition that 
\begin{align}
\langle \Psi, n_0 \Psi \rangle \geq c N,
\end{align}
for some $c>0$ depending only on the density and for all sufficiently large $N$.

We will study this problem for length scales $L$ that depend on the density $\rho$ in such a way that $L$ diverges as $\rho$ becomes small, i.e. as the system becomes dilute. We will prove that 
in this limit
\begin{align}\label{eq:SF_complete}
\frac{\langle \Psi, n_0 \Psi \rangle}{N} \rightarrow 1,\qquad \text{ or equivalently } \qquad 
\frac{\langle \Psi, n_{+} \Psi\rangle}{N} \rightarrow 0.
\end{align}

We will work under the following assumption on the potential.

\begin{assumption}\label{SF_assump:v}
The potential $v\neq 0$ is non-negative and spherically symmetric,
i.e. $v(x) = v(|x|)\geq 0$, and of class $L^1({\mathbb R}^3)$ with
compact support. We fix $R>0$ such that $\supp v \subset B(0,R)$.
\end{assumption}

Under Assumption~\ref{SF_assump:v} the potential $v$ has a finite scattering length $a>0$. 
For the convenience of the reader, we recall the definition of the scattering length in Appendix~\ref{scattering}.
The measure of diluteness is the (dimensionless) quantity $\rho a^3$, and we will study the limit $\rho a^3 \rightarrow 0$.
The length scale $L$ of the box is tied to the density $\rho$ in such a way that
\begin{align}
L = C_L (\rho a^3)^{-\delta} \frac{1}{\sqrt{\rho a}},
\end{align}
for some fixed constant $C_L>0$ and for $0<\delta$. The quantity $(\rho a)^{-1/2}$ has the unit of a length and is the Gross-Pitaevskii length scale mentioned above. Therefore, we see that we consider $L$'s much longer than the Gross-Pitaevskii length scale. For results about condensation on length scales shorter than $(\rho a)^{-1/2}$  see \cite{Rou} and references therein.

\begin{thm}\label{thm:SF_BECScales}
Suppose $v$ satisfies Assumption~\ref{SF_assump:v}.
Let $\epsilon \in [0,\frac{1}{2}]$ and $C_0 >0$.
Suppose that $\Psi \in \otimes_s^N L^2(\Omega)$ is a normalized state with
\begin{align}
\langle \Psi, H \Psi \rangle \leq 4\pi a \rho N + C_0  a \rho (\rho a^3)^{\frac{1}{2} - \epsilon} N.
\end{align}
Then,
\begin{align}
\frac{\langle \Psi, n_{+} \Psi\rangle}{N} \leq C \rho a L^2 (\rho a^3)^{\frac{1}{2} - \epsilon}.
\end{align}
In particular, if $2\delta + \epsilon < \frac{1}{2}$, then we have complete condensation in the sense of \eqref{eq:SF_complete}.
\end{thm}

\begin{rem}
On a box of side-length $L$ there is a gap in the spectrum of the kinetic energy above the ground state of order $L^{-2}$. We therefore, get the easy lower bound
\begin{align}
\langle \Psi, H \Psi \rangle \geq \langle \Psi, \widetilde{H} \Psi \rangle + 2\pi^2 L^{-2} \langle \Psi, n_{+} \Psi \rangle,
\end{align}
where
\begin{align}
\widetilde{H} := \sum_{j=1}^N \left(- \Delper_j - 2\pi^2 L^{-2} Q_{\Omega,j}\right) + \sum_{j < k } v^{\rm per}(x_j-x_k).
\end{align}
Therefore, Theorem~\ref{thm:SF_BECScales} follows upon proving the lower bound,
\begin{align}\label{eq:SF_LowerBound}
\widetilde{H} \geq 4\pi a \rho N + C_0  a \rho (\rho a^3)^{\frac{1}{2}} N.
\end{align}
The purpose of this paper is to show how 
this lower bound follows relatively easily using a localization argument from \cite{BFS,BS,FS} (see also \cite{LS1,LS2} for earlier versions of a similar argument), thereby proving condensation in the sense of Theorem~\ref{thm:SF_BECScales}.
This proof is very general, in particular, it is uniform in the $L^1$-norm of the potential, so it allows for hard core potentials. Also, the assumption on compact support can be relaxed. We will not pursue these generalizations---they follow fairly easily from the considerations in this paper---see \cite{BFS} for details.

Notice how this strategy depends strongly on the existence of a large enough gap in the kinetic energy. Therefore, it is clear that to prove BEC on much longer length scales other ideas will be needed. In particular, this technique does not work in the thermodynamic limit where one considers the limit $L\rightarrow \infty$ for fixed (small) density $\rho = N/L^3$.
\end{rem}

\begin{rem}
For large boxes\footnote{More precisely, the statement is valid in the thermodynamic limit.} we know that the energy per particle in the ground state has the form
\begin{align}\label{eq:SF_LHY}
E(N,L)/N = 4 \pi a \rho \left(1 + \frac{128}{15 \sqrt{\pi}} \sqrt{\rho a^3}\right) + {\mathcal O}(a \rho (\rho a^3)^{\frac{1}{2}+\eta}),
\end{align}
when $\rho a^3$ is small and for some $\eta>0$.
Here the leading order term was proved in \cite{dyson} (upper bound), \cite{LY} (lower bound), and the correction term---called the Lee-Huang-Yang-term \cite{LHY}---was rigorously established recently in \cite{YY} (upper bound)  and \cite{FS} (lower bound).

From the discussion in the previous remark it is clear that one can improve Theorem~\ref{thm:SF_BECScales} by going to the next order in the energy. More precisely, one can improve \eqref{eq:SF_LowerBound} by including the Lee-Huang-Yang term and calculating the energy per perticle to precision ${\mathcal O}(a \rho (\rho a^3)^{\frac{1}{2}+\eta})$.
This would allow us to establish BEC for states with sufficiently low energy on length scales a bit (depending on $\eta$) longer than $(\rho a^3)^{-\frac{1}{4}} \frac{1}{\sqrt{\rho a}}$.
For reasons of clarity of exposition we refrain from carrying out this more precise analysis in this short paper: The gain $\eta$ in \eqref{eq:SF_LHY} as proved in \cite{FS}
is not very large, and the additional necessary analysis is somewhat more lengthy than what is possible here. 
However, it should be clear from this paper, how to adapt  \cite{FS} to obtain the improved result.
\end{rem}

\begin{rem}
In the litterature (c.f. \cite{ABS,BBCS0,BBCS2}) the Hamiltonian is often written in another (equivalent) scaling, so it seems worthwhile to translate between the two. 
Consider the periodic box $[0,\lambda]^3$. Typically $\lambda$ is chosen to be equal to unity, but we keep it here for dimensional reasons.
Define $H^{\rm GP}(N)$ to be the operator
\begin{align}
H^{\rm GP}(N) :=  \sum_{j=1}^N - \Delper_j + \sum_{j < k } N^{2-2\kappa} v^{\rm per}(N^{1-\kappa}(y_j-y_k)),
\end{align}
where $\Delper$ is the Laplacian on the $\lambda$-torus.
By the scaling $x=N^{1-\kappa} y$, we see that
\begin{align}
H^{\rm GP}(N) \sim N^{2-2\kappa} H(N,L),
\end{align}
with $\sim$ denoting unitary equivalence and for $L= N^{1-\kappa} \lambda$.
We get the density after scaling as
$\rho = N L^{-3} = N^{3\kappa-2} \lambda^{-3}$, so in terms of the density we find
\begin{align}
L = C(\lambda)(\rho a^3)^{\frac{-\kappa}{4-6\kappa}} \frac{1}{\sqrt{\rho a}} ,
\end{align}
with $C(\lambda)$ being the constant $\left( \frac{a}{\lambda} \right)^{\frac{1}{2} + \frac{3\kappa}{4-6\kappa}}$.
Our results on BEC in the range $\delta \in (0,\frac{1}{4})$ when $\epsilon =0$, therefore translate into BEC for the low energy states of the operator $H^{\rm GP}(N)$, when $0<\kappa < 2/5$.
In \cite{ABS} a different method is used that gives BEC for $\kappa \in (0,\frac{1}{43})$.
\end{rem}

\section{Energy in small boxes}
We study a localized problem depending on a parameter $\rmu$. 
The following Section~\ref{localization} will show how the proof of Theorem~\ref{thm:SF_BECScales} reduces to this localized problem.
In the end $\rmu$ will be chosen to be the density $\rho = N/L^3$ in the original box, but we use a different notation to avoid confusion with the particle density of the localized problem.
The analysis in the present section follows closely that of \cite[Appendix B]{FS}, but with some simplifications similar to \cite{BFS}.
The box is $\Lambda = (-\frac{\ell}{2}, \frac{\ell}{2})^3$ of size
\begin{align}\label{eq:SF_def_ell}
\ell := K^{-1} \frac{1}{\sqrt{\rmu a}},
\end{align}
where $K>1$ will be chosen sufficiently large at the end, in order for the kinetic energy gap in the box to dominate various error terms (see \eqref{eq:SF_Gap} below).

We choose and fix a localization function $\chi \in C^{\infty}({\mathbb R}^3)$ with $\supp \chi \subset [-1/2,1/2]^3$.
We choose
$\chi$ even and satisfying
\begin{align}\label{eq:SF_chinormalization}
0 \leq \chi, \qquad \int \chi^2(x)\,dx = 1.
\end{align}
We will also use the notation 
\begin{align}
\chi_{\Lambda}(x) := \chi(x/\ell),
\end{align}
and
\begin{align}\label{eq:SF_3.5}
W(x) := \frac{v(x)}{\chi*\chi(x/\ell)}.
\end{align}
Since  $v$ has support in $B(0,R)$, we see that $W$ is well-defined if $\rmu a^3$ is sufficiently small.
Clearly $W$ depends on $\ell$ and thus $\rho_\mu$, but we will not reflect this in our notation.

Define furthermore, as operators on $L^2(\Lambda)$,
\begin{align}
P := \ell^{-3} | 1 \rangle \langle 1|, \qquad Q  := \one - P,
\end{align}
i.e. $P $ is the orthogonal projection in $L^2(\Lambda)$ onto the constant functions and $Q $ is the projection to the  orthogonal complement.

The localized Hamiltonian ${\mathcal H}_{\Lambda}$ acts on the symmetric Fock space ${\mathcal F}_s (L^2(\Lambda))$. It preserves particle number and is given as
\begin{align}\label{eq:SF_Def_HB}
({\mathcal H}_{\Lambda}(\rmu))_{M} :=
 \sum_{i=1}^M \mathcal{T}^{(i)} -
\rmu \sum_{i=1}^M \int w_{1}(x_i,y)\,dy + \sum_{1\leq i<j\leq M} w(x_i,x_j),
\end{align}
on the $M$-particle sector.
We will see in Section~\ref{localization} how ${\mathcal H}_{\Lambda}(\rmu)$ is related to the original $\widetilde{H}$ after localization through a sliding procedure.
Here
\begin{align}
\mathcal{T}  := Q\left[
\chi
\left( -\Delta- s^{-2}\ell^{-2} \right)_{+}
\chi
+
b \ell^{-2}   \right] Q,
\end{align}
with $s, b>0$ fixed constants, and
\begin{align}\label{eq:SF_w12}
w(x,y) &= \chi(x/\ell) W(x-y) \chi(y/\ell), \nonumber \\
W_1(x) &= W(x) (1-\omega(x)) = \frac{g(x)}{\chi*\chi(x/\ell)}, \nonumber \\
w_1(x,y) &= w(x,y) (1-\omega(x-y)) = \chi(x/\ell) g(x-y) \chi(y/\ell), 
\end{align}
where $\omega$ and $g$ are related to the scattering length and defined in \eqref{eq:SF_Scattering2}.
For later use, we also define
\begin{align}\label{eq:SF_w2}
w_2(x,y) &= w(x,y) (1-\omega^2(x-y)) = w_1(x,y) (1 + \omega(x-y)).
\end{align}
Notice the identities
\begin{align}\label{eq:SF_integrals}
\iint w_1(x,y)\,dx dy &= 8 \pi \ell^3 a, \nonumber \\
\iint w_2(x,y)\,dx dy &= \ell^3 \left( 8 \pi a + \int g \omega(x) \,dx\right).
\end{align}

We will prove the following lower bound on ${\mathcal H}_{\Lambda}(\rmu)$. Notice that this is a lower bound on the entire operator on Fock space, i.e. independent of particle number. This is useful because the term with $\rmu$ plays the role of a chemical potential and therefore effectively determines a preferred number of particles in the box.

\begin{thm}\label{thm:SF_LHY-Box}
Suppose $v$ satisfies Assumption~\ref{SF_assump:v} and let the localization function $\chi \in C^{\infty}_0((-\frac{1}{2}, \frac{1}{2})^3)$ be given.
Suppose that $K$ in \eqref{eq:SF_def_ell} is chosen sufficiently large (depending only on $\chi,b,s$ and the radius $R$ of  $\supp v$).
Then there exists a constant $C_0 >0$, such that for sufficiently small values of  $\rmu a^3$, we have that 
\begin{align}\label{eq:SF_EnergyBoxRes1}
{\mathcal H}_{\Lambda}(\rmu) \geq -4\pi \rmu^2 a \ell^3 -  C_0 \rmu^2 a  \ell^3 (\rmu a^3)^{\frac{1}{2}}.
\end{align}
\end{thm}

Theorem~\ref{thm:SF_LHY-Box} is actually the same as \cite[Theorem 6.1]{BFS}. For completeness, we give the main steps of the proof below.

\begin{proof}[Proof of Theorem~\ref{thm:SF_LHY-Box}]
We start by obtaining a control on the number of particles. Notice that ${\mathcal H}_{\Lambda}(\rmu)$ preserves particle number, so it suffices to establish the lower bound \eqref{eq:SF_EnergyBoxRes1} for $\langle \Psi, ({\mathcal H}_{\Lambda}(\rmu))_{M} \Psi \rangle$ for $\Psi \in \otimes_s^M L^2(\Lambda)$ for some arbitrary particle number $M$.
Since, for $\rmu a^3$ sufficiently small, $\rmu \ell^3 = K^{-3} (\rmu a^3)^{-1/2} \geq 1$, we can divide the $M$ particles in groups with the order of $\rmu \ell^3$ particles, i.e. given some $\Xi \geq 3$, we write
\begin{align}
\{1,\ldots,M\} = \cup_{j=1}^{\xi} S_j,
\end{align}
with $S_j \cap S_k = \emptyset$ for $j\neq k$ and 
\begin{align}\label{eq:SF_particleNumbers}
|S_{\xi}| &\leq (\Xi+1) \rmu \ell^3, \nonumber \\
|S_j| &\in [\Xi \rmu \ell^3,  (\Xi+1) \rmu \ell^3], \qquad \text{ for } j < \xi.
\end{align}

Now we remove the positive potentials $w(x_i,x_j)$ from $({\mathcal H}_{\Lambda}(\rmu))_{M}$ except when $i,j$ are in the same group $S_k$. i.e. discard the interaction between particles in different groups.
In this way, we obtain the lower bound,
\begin{align}\label{eq:SF_SumOverGroups}
\langle \Psi, ({\mathcal H}_{\Lambda}(\rmu))_{M} \Psi \rangle \geq
\sum_{j=1}^{\xi} \inf \Spec ({\mathcal H}_{\Lambda}(\rmu))_{|S_j|},
\end{align}
where the particle numbers $|S_j|$ satisfy \eqref{eq:SF_particleNumbers}.
Therefore, from now on we can reduce the analysis to $ ({\mathcal H}_{\Lambda}(\rmu))_{n}$ with $n \leq (\Xi+1) \rmu \ell^3$.

Define the operators $n_0, n_{+}$ on ${\mathcal F}_s (L^2(\Lambda))$, given by
\begin{align}
n_0 = \sum_{i} P_i,\qquad \qquad n_{+} = \sum_{i} Q_i.
\end{align}

An important step in the proof is the algebraic identity given in \eqref{eq:SF_potsplit} below, which isolates a positive term ${\mathcal Q}_4^{\rm ren}$ in the potential that can be discarded for a lower bound.

\begin{lem}[Potential energy decomposition]\label{lem:SF_potsplit}
We have
\begin{equation} \label{eq:SF_potsplit}
-\rmu \sum_{i=1}^n \int w_1(x_i,y)\,dy+
\frac{1}{2} \sum_{i\neq j}  w(x_i, x_j)
= {\mathcal Q}_0^{\rm ren}+{\mathcal Q}_1^{\rm ren}
+{\mathcal Q}_2^{\rm ren}+{\mathcal Q}_3^{\rm ren} + {\mathcal Q}_4^{\rm ren},
\end{equation}
where
\begin{align}\allowdisplaybreaks[4]
{\mathcal Q}_4^{\rm ren}:=&\,
\frac{1}{2} \sum_{i\neq j} \Big[ Q_i Q_j + (P_i P_j + P_i Q_j + Q_i P_j)\omega(x_i-x_j) \Big] w(x_i,x_j) \nonumber \\
&\,\qquad \qquad \times
\Big[ Q_j Q_i + \omega(x_i-x_j) (P_j P_i + P_j Q_i + Q_j P_i)\Big],\label{eq:SF_DefQ4}\\
{\mathcal Q}_3^{\rm ren}:=&\,
\sum_{i\neq j} P_i Q_j w_1(x_i,x_j) Q_j Q_i + h.c. \label{eq:SF_DefQ3} \\
{\mathcal Q}_2^{\rm ren}:=&\, \sum_{i\neq j} P_i Q_j w_2(x_i,x_j) P_j Q_i
+ \sum_{i\neq j} P_i Q_j w_2(x_i,x_j) Q_j P_i \nonumber \\&\,- \rmu \sum_{i=1}^n Q_i \int w_1(x_i,y)\,dy Q_i 
+ \frac{1}{2}\sum_{i\neq j} (P_i P_j w_1(x_i,x_j) Q_j Q_i + h.c.),\label{eq:SF_DefQ2}\\
{\mathcal Q}_1^{\rm ren}:=&\, \sum_{i,j}P_jQ_iw_2(x_i,x_j)P_iP_j-\rmu
  \sum_{i} Q_i \int w_1(x_i,y)\,dy P_i +h.c.
\label{eq:SF_DefQ1} \\
{\mathcal Q}_0^{\rm ren}:=&\,
\frac{1}{2} \sum_{i\neq j} P_i P_j w_2(x_i,x_j) P_j P_i - \rmu \sum_i P_i \int w_1(x_i,y)\,dy P_i\label{eq:SF_DefQ0}
\end{align}
\end{lem}
\begin{proof}
The identity \eqref{eq:SF_potsplit} is purely algebraic using that $P+Q=\one$, and the definitions \eqref{eq:SF_w12} and \eqref{eq:SF_w2}.
\end{proof}

\begin{rem}
The identity \eqref{eq:SF_potsplit} is used in \cite{FS} to get the energy to Lee-Huang-Yang precision. 
If one is only aiming to prove Theorem~\ref{thm:SF_LHY-Box} and is willing to let the error bound depend on $\int v$, then one can use the simpler
\begin{align*}
\widetilde{\mathcal Q}_4^{\rm ren}:=&\,
\frac{1}{2} \sum_{i\neq j} \Big[ Q_i Q_j + P_i P_j \omega(x_i-x_j) \Big] w(x_i,x_j) \Big[ Q_j Q_i + \omega(x_i-x_j) P_j P_i \Big].
\end{align*}
This will then, of course, change the other terms in \eqref{eq:SF_potsplit}.

Notice that one of the effects of the identity \eqref{eq:SF_potsplit} is that all terms with $3$ or fewer $Q$'s contain the potentials $w_1$ or $w_2$ instead of the original potential $v$. Integrals of $w_1, w_2$ are immediately controlled in terms of $a = \frac{1}{8\pi} \int g(x)\,dx$, which is possibly much smaller than $\frac{1}{8\pi} \int v(x)\,dx$. Therefore, using ${\mathcal Q}_4^{\rm ren}$ instead of $\widetilde{\mathcal Q}_4^{\rm ren}$ allows to get the uniformity of Theorem~\ref{thm:SF_LHY-Box}.
Therefore, in this article, we follow \cite{BFS,FS} and use the correct ${\mathcal Q}_4^{\rm ren}$-term.
\end{rem}

Using \eqref{eq:SF_potsplit} we can simplify the potential energy significantly. We can use that ${\mathcal Q}_4^{\rm ren}$ is a positive term to discard it. Also, we can use a Cauchy-Schwarz inequality on ${\mathcal Q}_3^{\rm ren}$ to absorb part of ${\mathcal Q}_3^{\rm ren}$ in the positive ${\mathcal Q}_4^{\rm ren}$-term (before discarding it). This will, of course, introduce additional terms with $2$ or fewer $Q$'s in them.
Furthermore, we notice that for all functions $F$,
$P_i F(x_i,y) P_i = \ell^{-3} P_i \int F(x,y) \,dx$, since $P$ projects on the constant function. 
This can be used to simplify several terms.
Finally, one can use a Cauchy-Schwarz inequality on the terms with one $Q$ to estimate these in terms of $2Q$-terms and $0Q$-terms.
This yields the following lemma, which is \cite[Lemma B.2]{FS} (where the details of the proof can also be found).

\begin{lem}[{\cite[Lemma B.2]{FS}}] \label{lem:SF_appinteractionestimate} There is a constant $C>0$ such that we have,
\begin{align}\label{eq:SF_SmallsimpleQs}
  -\rmu \sum_{i=1}^n \int w_1(x,y)\,dy+
   \frac{1}{2} \sum_{i\neq j}  w(x_i, x_j)
  \geq A_0+A_2-Ca (\rmu +n_0 \ell^{-3} )n_+  
\end{align}
where
\begin{align}
  A_0&=\frac{n_0(n_0-1)}{2\ell^6}\iint w_2(x,y)\,dx dy 
  - \left(\rmu \frac{n_0}{\ell^3}+\frac14\left(\rmu
  -\frac{n_0-1}{\ell^3}\right)^2\right)\iint w_1(x,y)\,dx dy
  \label{eq:SF_A0}
\end{align}
and 
\begin{equation}
  A_2= \frac{1}{2}\sum_{i\neq j} P_i P_j w_1(x_i,x_j) Q_j Q_i + h.c.
\end{equation} 
\end{lem}

Notice how, using the identities \eqref{eq:SF_integrals} and $n_0+n_+ = n$, we can get the following estimate
\begin{align}\label{lem:Simplified}
&-\rmu \sum_{i=1}^n \int w_1(x,y)\,dy+
   \frac{1}{2} \sum_{i\neq j}  w(x_i, x_j) \nonumber \\
&  \geq A_2 + \frac{n^2}{2\ell^3}\left(8\pi a + \int g \omega(x) \,dx\right) 
  -\left(\rmu \frac{n}{\ell^3}+\frac{1}{4}\left(\rmu
  -\frac{n}{\ell^3}\right)^2\right) 8\pi a \ell^3 \nonumber \\
&\quad  -Ca (\rmu +n \ell^{-3} )(n_+ +1)
\end{align}
from \eqref{eq:SF_SmallsimpleQs}.

Next we combine the kinetic energy (without the gap) and the $A_2$-term from Lemma~\ref{lem:SF_appinteractionestimate} in the following estimate. The analysis of this term uses second quantization.

\begin{lem}\label{lem:SF_Bog-Calc}
On the $n$-particle sector, we have (with ${\mathcal T}$ from \eqref{eq:SF_Def_HB} and $A_2$ from Lemma~\ref{lem:SF_appinteractionestimate})
\begin{align}\label{eq:SF_Bogoliubov}
\sum_{i=1}^n \left(\mathcal{T}^{(i)} - b \ell^{-2} Q_i\right) + A_2 &\geq  - \frac{n(n+1)}{2\ell^3} \int g(x) \omega(x)\,dx
- C a \frac{n_0+1}{\ell^3} n_{+} \nonumber \\
&\quad - C n a \ell^{-3} \left(1   +   \frac{a^2 (n+1)^2}{\ell^2} + \frac{(n+1)R^2}{\ell^2} \right).
\end{align}
\end{lem}

\begin{proof}
We define the following operators on ${\mathcal F}_s(L^2({\mathbb R}^3)$,
\begin{align}
b_k:= \ell^{-3/2} a_0^{\dagger} a(Q \chi_{\Lambda} e^{-ikx}), \qquad
b_k^{\dagger}:= \ell^{-3/2}  a^{\dagger}(Q \chi_{\Lambda} e^{-ikx}) a_0,
\end{align}
where $a, a^{\dagger}$ are the usual creation/annihilation operators on ${\mathcal F}_s(L^2({\mathbb R}^3))$, and $a_0 := \ell^{-3/2}a(\theta)$, with $\theta(x) = \one_{\Lambda}(x)$.
Notice for later use that, on the $n$-particle sector,
\begin{align}\label{eq:SF_commutator}
[ b_k, b_k^{\dagger}]  &= \ell^{-3}\left( a_0^{\dagger} a_0 \langle \chi_{\Lambda} e^{-ikx},  \chi_{\Lambda} e^{-ikx} \rangle
 - a^{\dagger}(Q \chi_{\Lambda} e^{-ikx}) a(Q \chi_{\Lambda} e^{-ikx}) \right) \nonumber \\
 &\leq n.
\end{align}

A direct calculation gives
\begin{align}
A_2 = \frac{1}{2} (2\pi)^{-3} \int \widehat{W_1}(p) \left( b_p^{\dagger} b_{-p}^{\dagger} + b_p b_{-p} \right) \,dp.
\end{align}
with $W_1$ from \eqref{eq:SF_w12}.

Furthermore, on the $n$-particle sector, we have, using that $n_0 \leq n$, that
\begin{align}
\sum_{i=1}^n \left(\mathcal{T}^{(i)} - b \ell^{-2} Q_i\right) \geq \frac{1}{2}  (2\pi)^{-3} \int \frac{\ell^3}{n+1}\tau(p) \left( b_p^{\dagger} b_{p}+ b_{-p}^{\dagger}  b_{-p} \right) \,dp,
\end{align}
with $\tau(p) = (p^2-s^{-2}\ell^{-2})_{+}$.

Finally, we add and subtract the term,
\begin{align}
\frac{1}{2} (2\pi)^{-3} \int 2 \widehat{W_1}(0) \left( b_p^{\dagger} b_{p}+ b_{-p}^{\dagger}  b_{-p} \right) \,dp.
\end{align}
Notice that
\begin{align}
\frac{1}{2} (2\pi)^{-3} \int 2 \widehat{W_1}(0) \left( b_p^{\dagger} b_{p}+ b_{-p}^{\dagger}  b_{-p} \right) \,dp \leq C a \frac{n_0+1}{\ell^3} n_{+},
\end{align}
so the subtracted term is in agreement with \eqref{eq:SF_Bogoliubov}.

We therefore estimate on the $n$-particle sector, 
\begin{align}
&\sum_{i=1}^n \left(\mathcal{T}^{(i)} - b \ell^{-2} Q_i\right) + A_2 + \frac{1}{2} (2\pi)^{-3} \int 2 \widehat{W_1}(0) \left( b_p^{\dagger} b_{p}+ b_{-p}^{\dagger}  b_{-p} \right) \,dp\nonumber\\
&\geq
\frac{1}{2} (2\pi)^{-3} \int {\mathcal A}(p) \left( b_p^{\dagger} b_{p}+ b_{-p}^{\dagger}  b_{-p} \right)  + \widehat{W_1}(p)   \left( b_p^{\dagger} b_{-p}^{\dagger}+ b_{-p}  b_{p} \right) \,dp,
\end{align}
where
\begin{align}
{\mathcal A}(p) :=  \frac{\ell^3}{n+1}\tau(p) + 2 \widehat{W_1}(0).
\end{align}
Upon calculating a square $(\alpha b_p^{\dagger} + \beta b_{-p}) (\alpha b_p + \beta b_{-p}^{\dagger}) + (\alpha b_{-p}^{\dagger} + \beta b_{p}) (\alpha b_{-p} + \beta b_{p}^{\dagger})$ with appropriately adjusted coefficients $\alpha, \beta$---and afterwards discarding the square as being positive---we get the estimate (see \cite[Appendix A]{FS} for details)
\begin{align}
& {\mathcal A}(p) \left( b_p^{\dagger} b_{p}+ b_{-p}^{\dagger}  b_{-p} \right)  + \widehat{W_1}(p)   \left( b_p^{\dagger} b_{-p}^{\dagger}+ b_{-p}  b_{p} \right)\nonumber \\
& \geq
 -\frac{1}{2} \left({\mathcal A}(p) - \sqrt{{\mathcal A}(p)^2- \widehat{W_1}(p)^2}\right) \left( [ b_p , b_p^{\dagger} ] + [b_{-p} , b_{-p}^{\dagger} ]\right) \nonumber \\
 & \geq
 - n  \left({\mathcal A}(p) - \sqrt{{\mathcal A}(p)^2- \widehat{W_1}(p)^2}\right),
\end{align}
where the last inequality uses \eqref{eq:SF_commutator}.
Notice here how the added term $2  \widehat{W_1}(0)$ ensures that ${\mathcal A}(p)^2- \widehat{W_1}(p)^2\geq 0$.
Using again that $\frac{|\widehat{W_1}(p)|}{{\mathcal A}(p)} \leq \frac{1}{2}$ to expand the square root, we find
\begin{align}\label{eq:SF_BogCalc}
& \sum_{i=1}^n \left(\mathcal{T}^{(i)} - b \ell^{-2} Q_i\right) + A_2 + \frac{1}{2} (2\pi)^{-3} \int 2 \widehat{W_1}(0) \left( b_p^{\dagger} b_{p}+ b_{-p}^{\dagger}  b_{-p} \right) \,dp \nonumber \\
& \geq -\frac{n}{2} (2\pi)^{-3} \int 
  \frac{ \widehat{W_1}(p)^2}{2 {\mathcal A}(p)} \,dp -  Cn \int \frac{ \widehat{W_1}(0)^4}{2 {\mathcal A}(p)^3} \, dp.
\end{align}
To estimate the last integral, we split in $\{|p|^2 \geq 2 s^{-2} \ell^{-2}\} \cup \{|p|^2 \leq 2 s^{-2} \ell^{-2}\} $.
For the large values of $|p|$ we have $ {\mathcal A}(p) \geq \frac{\ell^3}{2(n+1)} p^2$. Also, for all values of $|p|$, we have $ {\mathcal A}(p) \geq 2 \widehat{W}(0)$. With this, it is easy to control the integral in $p$ and get
\begin{align}
n \int \frac{ \widehat{W_1}(0)^4}{2 {\mathcal A}(p)^3} \, dp \leq 
n \int \frac{ \widehat{W_1}(0)^3}{2 {\mathcal A}(p)^2} \, dp \leq 
C n \left( \left(\frac{n+1}{\ell^3}\right)^2 a^3 \ell + a \ell^{-3}\right),
\end{align}
in agreement with the error bound in \eqref{eq:SF_Bogoliubov}.

The first integral in \eqref{eq:SF_BogCalc} has to be calculated. We write it as
\begin{align}
-\frac{n}{2} (2\pi)^{-3} \int 
  \frac{ \widehat{W_1}(p)^2}{2 {\mathcal A}(p)} \,dp = I + II,
\end{align}
with
\begin{align}
I:= -\frac{n(n+1)}{2\ell^3} (2\pi)^{-3} \int 
  \frac{ \widehat{W_1}(p)^2}{2  p^2 } \,dp,
\end{align}
\begin{align}
II:= -\frac{n}{2} (2\pi)^{-3} \int  \frac{ \widehat{W_1}(p)^2}{2 {\mathcal A}(p)} -
  \frac{ \widehat{W_1}(p)^2}{2 \frac{\ell^3}{n+1} p^2 } \,dp.
\end{align}
In $II$, we again estimate by splitting the integral as above. Notice how $\tau(p) = p^2-s^{-2}\ell^{-2}$ when $|p|^2 \geq 2 s^{-2} \ell^{-2}$, so we easily get
\begin{align}
| II | &\leq C n a^2 \left( \int_{\{|p|^2 \leq 2 s^{-2} \ell^{-2}\} } a^{-1} + \frac{n+1}{\ell^3 p^2}\,dp +
\int_{\{|p|^2 \geq 2 s^{-2}  \ell^{-2}\} } \frac{ \frac{\ell^3}{n+1} s^{-2} \ell^{-2} + 2 \widehat{W_1}(0)}{(\frac{\ell^3}{n+1} )^2 p^4} \,dp\right) \nonumber \\
&\leq C n \left( a \ell^{-3} + a^3 \ell \left(\frac{n+1}{\ell^3}\right)^2 \right).
\end{align}
This is also a bound that can be absorbed in the error term in \eqref{eq:SF_Bogoliubov}.

Finally, we consider the integral $I$. Notice that $0 \leq W_1(x) \leq (1+C(R/\ell)^2) g(x)$, using \eqref{eq:SF_w12} and a simple bound on the convolution (recall that $\supp v \subset B(0,R)$).
Therefore,
\begin{align}
(2\pi)^{-3} \int 
  \frac{ \widehat{W_1}(p)^2}{2  p^2 } \,dp
  &= \iint \frac{W_1(x) W_1(y)}{2|x-y|} \,dx dy \nonumber \\
  &\leq (1+C(R/\ell)^2) (2\pi)^{-3} \int \frac{\widehat{g}(p)^2}{2p^2}\,dp \nonumber \\
  &= (1+C(R/\ell)^2) \int g(x) \omega(x)\,dx,
\end{align}
where we used \eqref{es:scatteringFourier} to get the last identity.
Upon inserting this in the expression for $I$,
this gives the main term in \eqref{eq:SF_Bogoliubov} as well as the error term depending on $R^2$.

This finishes the proof of Lemma~\ref{lem:SF_Bog-Calc}.
\end{proof}

Using Lemma~\ref{lem:SF_appinteractionestimate} (in the form of \eqref{lem:Simplified})
and Lemma~\ref{lem:SF_Bog-Calc}, we can now finish the proof of Theorem~\ref{thm:SF_LHY-Box}.
Combining the terms, remembering that the \enquote{gap} of the kinetic energy was saved, we find on the $n$-particle subspace,
\begin{align}
{\mathcal H}_{\Lambda}(\rmu))_{n} \geq
E_{{\rm Main}} + E_{{\rm gap}} + E_{{\rm error}},
\end{align}
with
\begin{align}
E_{{\rm Main}} &:=
\frac{n^2}{2\ell^3}8\pi a   -\left(\rmu \frac{n}{\ell^3}+\frac{1}{4}\left(\rmu
  -\frac{n}{\ell^3}\right)^2\right) 8\pi a \ell^3 \nonumber \\
  &= - 4\pi a \rmu^2 \ell^3 + 2\pi \frac{a}{\ell^3} (\rmu \ell^3 - n)^2,
\end{align}
as well as
\begin{align}\label{eq:SF_Gap}
E_{{\rm gap}} := \left( b \ell^{-2} - C a(\frac{n+1}{\ell^3} + \rmu) \right) n_{+},
\end{align}
and
\begin{align}
E_{{\rm error}} := - C n a \ell^{-3} \left(1   +   \frac{a^2 (n+1)^2}{\ell^2} + \frac{(n+1)R^2}{\ell^2} \right)-Ca \rmu.
\end{align}
When $n \leq (\Xi+1) \rmu \ell^3$, we find
\begin{align}
E_{{\rm error}} \geq - C \rmu a\left(1 + (R/a)^2 \sqrt{\rmu a^3}\right),
\end{align}
with $C$ depending on $\Xi$ and on $K$ (from the definition \eqref{eq:SF_def_ell} of $\ell$).

Also for $n \leq (\Xi+1) \rmu \ell^3$, we may choose $K$ large enough (depending on $\Xi$) to get
\begin{align}
E_{{\rm gap}} \geq \frac{b}{2\ell^2} n_{+} \geq 0.
\end{align}
In conclusion, we get, for $n \leq (\Xi+1) \rmu \ell^3$,
\begin{align}\label{eq:SF_FinalEnergy}
{\mathcal H}_{\Lambda}(\rmu))_{n} \geq - 4\pi a \rmu^2 \ell^3 + 2\pi \frac{a}{\ell^3} (\rmu \ell^3 - n)^2 - C \rmu a\left(1 + (R/a)^2 \sqrt{\rmu a^3}\right).
\end{align}
In particular, for $n \in [\Xi \rmu \ell^3,  (\Xi+1) \rmu \ell^3]$, and using that $\Xi \geq 3$,
\begin{align}
{\mathcal H}_{\Lambda}(\rmu))_{n} \geq \rmu a \left( 4\pi \rmu \ell^3 - C \left(1 + (R/a)^2 \sqrt{\rmu a^3}\right)\right) \geq 0,
\end{align}
for $\rmu a^3$ sufficiently small.

Upon inserting this as well as \eqref{eq:SF_FinalEnergy} (for the last group of particles) in \eqref{eq:SF_SumOverGroups}, we get the bound \eqref{eq:SF_EnergyBoxRes1}.

This finishes the proof of Theorem~\ref{thm:SF_LHY-Box}.
\end{proof}

\section{Localization to small boxes}\label{localization}

To prove \eqref{eq:SF_LowerBound} we will localize to smaller boxes. In order not to have to worry about how the particles distribute themselves between the boxes it is convenient to 
reformulate our problem on Fock space. 
Consider, for given $\rmu >0$, the following operator ${\mathcal
  H}_{\rmu}$ on the symmetric Fock space ${\mathcal F}_{\rm
  s}(L^2(\Omega))$. The operator ${\mathcal H}_{\rmu}$ commutes with
particle number and satisfies, with ${\mathcal H}_{\rmu,N}$ denoting
the restriction of ${\mathcal H}_{\rmu}$ to the $N$-particle subspace
of ${\mathcal F}_{\rm s}(L^2(\Omega))$,
\begin{align}\label{eq:SF_BackgroundH}
{\mathcal H}_{\rmu,N} &=\sum_{i=1}^N \left(- \Delta_i - 2\pi^2 L^{-2} Q_i\right)
+ \sum_{i<j} v^{\rm per}(x_i-x_j)
- 8\pi a \rmu N\\
&=\sum_{i=1}^N \left(- \Delta_i - 2\pi^2 L^{-2} Q_i - \rmu \int_{{\mathbb R}^3} g(x_i-y)\,dy \right)
+ \sum_{i<j} v^{\rm per}(x_i-x_j)
 \nonumber 
\end{align}
where $g$ is defined in terms of the scattering solution (see \eqref{eq:SF_gdef}).

We will prove the following lower bound on ${\mathcal H}_{\rmu}$.

\begin{thm}\label{thm:SF_LHY-Background}
Suppose that $v$ satisfies Assumption~\ref{SF_assump:v}.
Then there exists $C_1>0$ such that
\begin{align}
  {\mathcal H}_{\rmu} \geq -4\pi \rmu^2 a L^3 \left(1 + C_1 (\rmu a^3)^{1/2}\right),
\end{align}
for all $\rmu a^3$ sufficiently small.
\end{thm}

Before proving Theorem~\ref{thm:SF_LHY-Background} we show that \eqref{eq:SF_LowerBound}---and therefore our main result Theorem~\ref{thm:SF_BECScales}---follows from it.

\begin{proof}[Proof of  \eqref{eq:SF_LowerBound} using Theorem~\ref{thm:SF_LHY-Background}]
To prove  \eqref{eq:SF_LowerBound}, let $\Psi \in \otimes_s^N L^2(\Omega) \subset {\mathcal F}_{\rm s}(L^2(\Omega))$ be normalized. Then, using Theorem~\ref{thm:SF_LHY-Background},
\begin{align}
\langle \Psi, \widetilde{H} \Psi \rangle &\geq \langle \Psi, {\mathcal
  H}_{\rmu} \Psi \rangle + 8 \pi a \rmu \rho L^3 \nonumber \\
&  \geq 4\pi a \rho^2 L^3 - 4\pi a (\rho - \rmu)^2 L^3 - 4\pi \rmu^2 a L^3C_1 (\rmu a^3)^{1/2}.
\end{align}
The estimate \eqref{eq:SF_LowerBound} follows upon choosing $\rmu = \rho$.
\end{proof}

\begin{proof}[Proof of Theorem~\ref{thm:SF_LHY-Background}]
We will localize to boxes of size $\ell$, with $\ell$ being as defined in \eqref{eq:SF_def_ell}.
In particular, it is important that $R < \ell < \frac{1}{2}L$.
It turns out that the Hamiltonians on the small boxes are all unitarily equivalent, so we really only have one box $\Lambda=[-\ell/2,\ell/2]^3$ and the Hamiltonian on this box can be identified with ${\mathcal H}_{\Lambda}(\rmu)$ defined in \eqref{eq:SF_Def_HB}. Therefore, Theorem~\ref{thm:SF_LHY-Background} will be a consequence of Theorem~\ref{thm:SF_LHY-Box}. Below we give the details of this argument.

For given $u \in {\mathbb R}^3$,
define the sharp localization function $\theta_u$ to the box $\Lambda(u) := u + [-\ell/2,\ell/2]^3$, i.e.
\begin{align}\label{eq:SF_Theta}
\theta_u := \one_{\Lambda(u)}.
\end{align}
The function $\theta_u$ defines an orthogonal projection $\theta_u: L^2(\Omega) \rightarrow L^2(\Lambda(u))$, where we recall that $\Omega$ is identified with the $L$-torus.
Similarly, define $P_u, Q_u$ to be the maps $P_u, Q_u: L^2(\Omega) \rightarrow L^2(\Lambda(u))$, defined by
\begin{align}\label{def: projections}
P_u \varphi := \ell^{-3} \langle \theta_u, \varphi\rangle \theta_u, \qquad  Q_u \varphi:= \theta_u \varphi - \ell^{-3} \langle \theta_u, \varphi \rangle \theta_u.
\end{align}
We also define, for $u \in {\mathbb R}^3$,
\begin{align}
\chi_u(x) = \chi(\frac{x-u}{\ell})=\chi_\Lambda(x-u).
\end{align}

Recall the definition of $W$  from \eqref{eq:SF_3.5}.
We will also need the periodic versions
\begin{align}
W^{\rm per}(x) = \sum_{j \in {\mathbb Z}^3} W(x+Lj),
\end{align}
and similarly for $\chi_u^{\rm per}$.

Define the localized potentials
\begin{align}\label{eq:SF_w_u}
w_u(x,y) := \chi_u(x) W(x-y) \chi_u(y), \qquad \text{ so that } \qquad w(x,y) = w_{u=0}(x,y),
\end{align}
and
\begin{align}\label{eq:SF_w_u-per}
w_u^{\rm per}(x,y) := \chi_u^{\rm per}(x) W^{\rm per}(x-y) \chi_u^{\rm per}(y).
\end{align}

We also recall the functions $W_1$ and $w_1$ defined in \eqref{eq:SF_w12}.
If we add a subscript $u$ we mean as above the translated versions $w_{1,u}(x,y)=w_1(x-\ell u,y-\ell u)$, and similarly we add a superscript `${\rm per}$` for the periodic versions as in \eqref{eq:SF_w_u-per}.

We then get by a direct calculation:
\begin{lem}\label{lem:LocPotEn}
For $2\ell <L$, we get
\begin{align}
  -\rmu \sum_{i=1}^N &\int  g(x_i-y)\,dy + \sum_{i<j} v^{\rm per}(x_i-x_j)    =
   \ell^{-3}\int_{\Omega} 
  {\mathcal W}_{u,N}\,du,
\end{align}
with
\begin{align}
   {\mathcal W}_{u,N}:=
-\rmu \sum_{i=1}^N \int_{\Omega} w_{1,u}^{\rm per}(x_i,y) \,dy + \sum_{i<j}  w_u^{\rm per}(x_i,x_j).
\end{align}
\end{lem}

\begin{proof}
The proof follows by direct calculation of the $u$-integral.
\end{proof}

We will estimate the kinetic energy $-\Delper$ in $\Omega$ below by an
integral over kinetic energy operators in the boxes $\Lambda(u)$. The
following theorem is essentially \cite[Lemma 5.7]{BFS}.

\begin{lem}[Kinetic energy localizaton]\label{lem:LocKinEn} 
Assume that $2\ell < L$.
There exists a constant $b>0$, such that if $s$ is small enough, then 
\begin{equation}
 \ell^{-3} \int_{\Omega} {\mathcal T}_u  \,du\leq -\Delper - 2\pi L^{-2} Q_L,
\end{equation}
where $\Delper$ denotes the Laplacian on $\Omega$ with periodic boundary conditions, and
where
\begin{align}\label{eq:SF_DefTu}
\mathcal{T}_u  := Q_u^{*}\left[
\chi_u 
\left( -\Delta- s^{-2}\ell^{-2} \right)_{+}
\chi_u
+
b \ell^{-2}   \right] Q_u,
\end{align}
where $\Delta$ is the Laplacian on ${\mathbb R}^3$.
\end{lem}
We leave the proof of Lemma~\ref{lem:LocKinEn} to the end of the section, and finish the proof of Theorem~\ref{thm:SF_LHY-Background} first.

By combining Lemmas~\ref{lem:LocKinEn} and \ref{lem:LocPotEn}, we see that
\begin{align}\label{eq:SF_sliding}
{\mathcal H}_{\rmu,N} \geq \ell^{-3}\int_{\Omega} ({\mathcal H}_{\Lambda(u)}(\rmu))_{N}\,du,
\end{align}
with
\begin{align}
({\mathcal H}_{\Lambda(u)}(\rmu))_{N}= \sum_{i=1}^N  \mathcal{T}^{(i)}_u+ {\mathcal W}_{u,N}.
\end{align}
Consider now the decomposition $L^2(\Omega) = L^2(\Lambda(u)) \oplus L^2(\Omega\setminus \Lambda(u))$.
Using that ${\mathcal F}_s(L^2(\Omega))$ is unitarily equivalent to ${\mathcal F}_s(L^2(\Lambda(u))) \otimes {\mathcal F}_s(L^2(\Omega\setminus \Lambda(u)))$ it easy to see that
\begin{align}
\inf \Spec {\mathcal H}_{\Lambda(u)}(\rmu) = \inf \Spec \widetilde{{\mathcal H}}_{\Lambda(u)}(\rmu),
\end{align}
where $\widetilde{{\mathcal H}}_{\Lambda(u)}(\rmu)$ denotes the operator $ {\mathcal H}_{\Lambda(u)}(\rmu)$ considered on ${\mathcal F}_s(L^2(\Lambda(u)))$.
Upon identifying $L^2(\Lambda(u)) \subset L^2(\Omega)$ with $L^2(\Lambda(u))\subset L^2({\mathbb R}^3)$, it is now clear that all the 
$\widetilde{{\mathcal H}}_{\Lambda(u)}(\rmu)$ are pairwise unitarily equivalent---in particular, they are all unitarily equivalent to ${\mathcal H}_{\Lambda}(\rmu)$ defined in \eqref{eq:SF_Def_HB}. Therefore, by Theorem~\ref{thm:SF_LHY-Box} we have
\begin{align}
({\mathcal H}_{\Lambda(u)}(\rmu))_{N} \geq -4\pi \rmu^2 a \ell^3 -  C_0 \rmu^2 a  \ell^3 (\rmu a^3)^{\frac{1}{2}},
\end{align}
for all $N$ and all $u \in \Omega$.
Upon inserting this bound in \eqref{eq:SF_sliding}, we get the desired bound on ${\mathcal H}_{\rmu}$.
This finishes the proof of Theorem~\ref{thm:SF_LHY-Background}.
\end{proof}

\begin{proof}[Proof of Lemma~\ref{lem:LocKinEn}]
Define, for $p \in {\mathbb R}^3$,
\begin{align}
\tau(p) := (p^2-s^{-2}\ell^{-2})_{+}.
\end{align}
An explicit calculation shows that for all $k,k' \in 2\pi L^{-1} {\mathbb Z}^3$,
\begin{align}
L^{-3} \int_{\Omega} \langle e^{ikx}, {\mathcal T}_u e^{ik'x} \rangle \,du =
\delta_{k,k'} \ell^3 F(k),
\end{align}
with
\begin{align}
F(k) = \ell^{-3}  \langle e^{ikx}, {\mathcal T}_{u=0} e^{ikx} \rangle.
\end{align}
So we have to show that, if $b,s$ are small enough, then for all $k \in 2 \pi L^{-1} {\mathbb Z}^3\setminus \{0\}$,
\begin{align}\label{eq:SF_k2}
F(k) \leq k^2 - 2 \pi^2 L^{-2}.
\end{align}
It is easy to reduce by scaling to the case $\ell = 1$.
We write, now with the convention that $\ell =1$,
$$
F(k) = F_1(k) + F_2(k),
$$
where (using that $\chi$ is even to avoid complex conjugates) 
\begin{align}
F_1= I + II + III,\qquad 
F_2= b \left( 1 - \widehat{\theta}(k)^2\right),
\end{align}
and 
\begin{align}
I &:= (2\pi)^{-3} (\tau * (\widehat{\chi_{\Lambda}}^2)(k), \nonumber \\
II &:= -2 (2\pi)^{-3}  [ (\tau \widehat{\chi_{\Lambda}})*\widehat{\chi_{\Lambda}}](k) \widehat{\theta}(k),\nonumber \\
III&:=  (2\pi)^{-3}  \widehat{\theta}(k)^2 (\tau * (\widehat{\chi_{\Lambda}}^2)(0), 
\end{align}

An explicit calculation of the Fourier transform of $\theta$ easily yields that
\begin{align}\label{eq:SF_F2}
 0 \leq F_2(k) \leq b \min\{  \beta k^2, 1\},
\end{align}
for suitable $\beta>0$.

The estimate on $F_1$ is divided in two according to the size of $|k|$. For $|k| < \frac{1}{2} s^{-1}$ we prove that
\begin{align}\label{eq:SF_Taylor}
0 \leq F_1(k) \leq C s k^2.
\end{align}
It is clear that $F_1(0)=0$ and $\nabla_k F_1(0)=0$ follows by parity. Therefore, by Taylor's Theorem it suffices to prove that
$|\partial_{k_i} \partial_{k_j} F_1| \leq C s$. This is done by calculation on each of the three terms, using the decay and regularity of $\chi_{\Lambda}$.
For example, for the term $III$ we have 
$|\partial_{k_i} \partial_{k_j}  \widehat{\theta}(k)^2| \leq C$, since $\theta$ is supported in the unit cube, and
$(\tau * (\widehat{\chi_{\Lambda}}^2)(0) \leq  \int_{\{|p| \geq s^{-1}\}} p^2 (\widehat{\chi_{\Lambda}}(p))^2\,dp \leq C s$, by the fast decay of $\widehat{\chi_{\Lambda}}$, which follows from its regularity. The restriction that $|k| \leq  \frac{1}{2} s^{-1}$ is used in the estimate of $I$ to ensure that $(p-k)^2 - s^{-2} \geq 0$ implies that $|p| \geq \frac{1}{2} s^{-1}$ and therefore giving decay of $\widehat{\chi_{\Lambda}}$. We leave the remaining details to the reader.

Using \eqref{eq:SF_Taylor} and \eqref{eq:SF_F2}, we get \eqref{eq:SF_k2} for $|k| \leq  \frac{1}{2} s^{-1}$ as soon as $b \beta + C s \leq \frac{1}{2}$ (with $C$ from \eqref{eq:SF_Taylor} and $\beta$ from \eqref{eq:SF_F2}) and recalling that $k^2 \geq 4 \pi^2 L^{-2}$ for all allowed $k$.

For $|k| \geq  \frac{1}{2} s^{-1}$, we estimate the three terms in $F_1$ individually.
Since $\chi$ is smooth, it is easy to see that
\begin{align}\label{eq:SF_2+3}
 | II | + | III | \leq C.
\end{align}
We estimate $I$ with $\varepsilon \in (0,1)$, and using $[x+y]_{+} \leq [x]_+ + [y]_+$ and that $\chi$ is even, as
\begin{align}
I&= (2\pi)^{-3} (\tau * (\widehat{\chi_{\Lambda}}^2)(k) \nonumber \\
&= 
(2\pi)^{-3}\int ((p-k)^2 - s^{-2}) \widehat{\chi_{\Lambda}}^2(p)
+
 [s^{-2}-p^2 + 2kp - k^2]_{+}  \widehat{\chi_{\Lambda}}^2(p)\,dp \nonumber \\
&\leq
k^2 -s^{-2} + [s^{-2} - (1 -\varepsilon) k^2]_{+} + \varepsilon^{-1} (2\pi)^{-3}  \int p^2 \widehat{\chi_{\Lambda}}^2(p) \,dp.
\end{align}
Therefore, for $|k| \geq \frac{1}{2} s^{-1}$ we get
\begin{align}
 I \leq k^2 -(1-\varepsilon) \frac{1}{4}  s^{-2}
+\varepsilon^{-1}  \int | \nabla \chi |^2\,dx.
\end{align}
Upon fixing $\varepsilon= \frac{1}{2}$ this becomes
\begin{align}\label{eq:SF_LeadingI}
 I \leq k^2 -  \frac{1}{8} s^{-2}
+C , 
\end{align}
for all $|k| \geq \frac{1}{2} s^{-1}$ and with $C$ independent of $s$. 

Upon adding \eqref{eq:SF_LeadingI}, \eqref{eq:SF_F2} and \eqref{eq:SF_2+3}, see that \eqref{eq:SF_k2} holds for $|k| \geq \frac{1}{2} s^{-1}$ for all fixed $b>0$, as soon as $s$ is small enough.
\end{proof}

\appendix 
\section{Facts about the scattering solution}\label{scattering}
We briefly recall the definition of the scattering length and related quantities. For further details see \cite[Appendix C]{LSSY}.

We assume that $v$ satisfies Assumption~\ref{SF_assump:v}.
The equation
\begin{align}\label{eq:SF_Scattering2}
  (-\Delta + \frac{1}{2} v(x) )(1-\omega(x)) =0,\qquad \text{ with } \omega \rightarrow 0, \text{ as } |x| \rightarrow \infty,
\end{align}
is called the scattering equation. 
The radial solution $\omega$ to this equation satisfies that $\omega(x) = a/|x|$ for some $a>0$ and for $x$ outside
$\supp\, v$. This constant $a$ is the {\it scattering length} of the
potential $v$.
The function $\omega$ is radially symmetric and
non-increasing with
\begin{align}
	0\leq \omega(x)\leq 1.\label{omegabounds}
\end{align}
We can rewrite \eqref{eq:SF_Scattering2} as
\begin{align}
\label{eq:SF_Scattering3}
-\Delta \omega = \frac{1}{2} g,
\end{align}
with
\begin{align}\label{eq:SF_gdef}
g := v(1-\omega),
\end{align}
and get
\begin{align}
a = (8\pi)^{-1} \int g(x)\,dx,
\end{align}
and that the Fourier transform satisfies
\begin{align}\label{es:scatteringFourier}
\widehat{\omega}(k) = \frac{\hat{g}(k)}{2 k^2}.
\end{align}

\begin{acknowledgments}
This research was partly supported by the Charles Simonyi Endowment and by an EliteResearch Prize from the Danish Ministry
of Higher Education and Science.
\end{acknowledgments}

\end{document}